\documentclass[letterpaper, 10 pt, conference]{ieeeconf}  

\IEEEoverridecommandlockouts                              

\usepackage{amsfonts, amsmath, amssymb, dsfont}
\usepackage[nocomma,short]{optidef}
\usepackage{enumerate}
\usepackage{booktabs}
\usepackage[cmintegrals]{newtxmath}
\usepackage{xcolor}
\let\labelindent\relax
\usepackage{enumitem}
\usepackage{array}

\DeclareMathOperator*{\argmin}{arg\,min}
\DeclareMathOperator*{\argmax}{arg\,max}

\newtheorem{theorem}{\textbf{Theorem}}
\newtheorem{corollary}{\textbf{Corollary}}

\newtheorem{lemma}{\textbf{Lemma}}

\usepackage[pscoord]{eso-pic}
\newcommand{\placetextbox}[3]{
  \setbox0=\hbox{#3}
  \AddToShipoutPictureFG*{
    \put(\LenToUnit{#1\paperwidth},\LenToUnit{#2\paperheight}){\vtop{{\null}\makebox[0pt][c]{#3}}}%
  }%
}%

\begin{document}
\placetextbox{0.5}{0.97}{\texttt{This work has been submitted to the IEEE for possible publication. Copyright may be }}%
\placetextbox{0.5}{0.954}{\texttt{transferred without notice, after which this version may no longer be accessible.}}%
%
\title{Multi-Attribute Auctions for Efficient Operation of Non-Cooperative Relaying Systems}
%
%
%

\author{Winston~Hurst,~\IEEEmembership{Student Member,~IEEE,}
        and~Yasamin~Mostofi,~\IEEEmembership{Fellow,~IEEE}
\thanks{Winston Hurst and Yasamin Mostofi are with the Department of Electrical and Computer Engineering, University of California, Santa Barbara, USA (email: \{winstonhurst, ymostofi\}@ece.ucsb.edu). This work was supported in part by ONR award N00014-23-1-2715.}
}

\maketitle

\begin{abstract}
This paper studies the use of a multi-attribute auction in a communication system to bring about efficient relaying in a non-cooperative setting. We consider a system where a source seeks to offload data to an access point (AP) while balancing both the timeliness and energy-efficiency of the transmission. A deep fade in the communication channel (due to, e.g., a line-of-sight blockage) makes direct communication costly, and the source may alternatively rely on non-cooperative UEs to act as relays. We propose a multi-attribute auction to select a UE and to determine the duration and power of the transmission, with payments to the UE taking the form of energy sent via wireless power transfer (WPT). The quality of the channel from a UE to the AP constitutes private information, and bids consist of a transmission time and transmission power. We show that under a second-preferred-offer auction, truthful bidding by all candidate UEs forms a Nash Equilibrium. However, this auction is not incentive compatible, and we present a modified auction in which truthful bidding is in fact a dominant strategy. Extensive numerical experimentation illustrates the efficacy of our approach, which we compare to a cooperative baseline. We demonstrate that with as few as two candidates, our improved mechanism leads to as much as a $76\,$\% reduction in energy consumption, and that with as few as three candidates, the transmission time decreases by as much as $55\,$\%. Further, we see that as the number of candidates increases, the performance of our mechanism approaches that of the cooperative baseline. Overall, our findings highlight the potential of multi-attribute auctions to enhance the efficiency of data transfer in non-cooperative settings.
\end{abstract}


%
\IEEEpeerreviewmaketitle

\section{Introduction}\label{sec:introduction}
The realization of the 6G vision for communication networks requires the incorporation of new technologies and paradigms. For example, high-frequency mmWave and THz communications will greatly increase capacity, while densely connected and intelligent network topologies will improve energy and bandwidth efficiency through dynamic, distributed decision making~\cite{IMT-2030}. 

In this context, this paper considers a source which must send data to an access point (AP) while balancing time and energy considerations. We assume that the direct channel between the source and AP is of poor quality, due to, e.g., acute penetration loss caused by a line-of-sight (LOS) blockage on a mmWave channel \cite{Access2019Rappaport}. Rather than send the data directly, the source may transfer the data via one of several battery-powered candidate user equipments (UEs). Each candidate is modeled as non-cooperative, rational entity with private information about the channel quality between itself and the AP, and to motivate cooperation, the source must provide a ``payment" in the form of energy using wireless power transfer (WPT).

\begin{figure}
    \centering
    \includegraphics[width=0.9\linewidth, trim={0.1in, 0.1in, 0in, 0in}, clip]{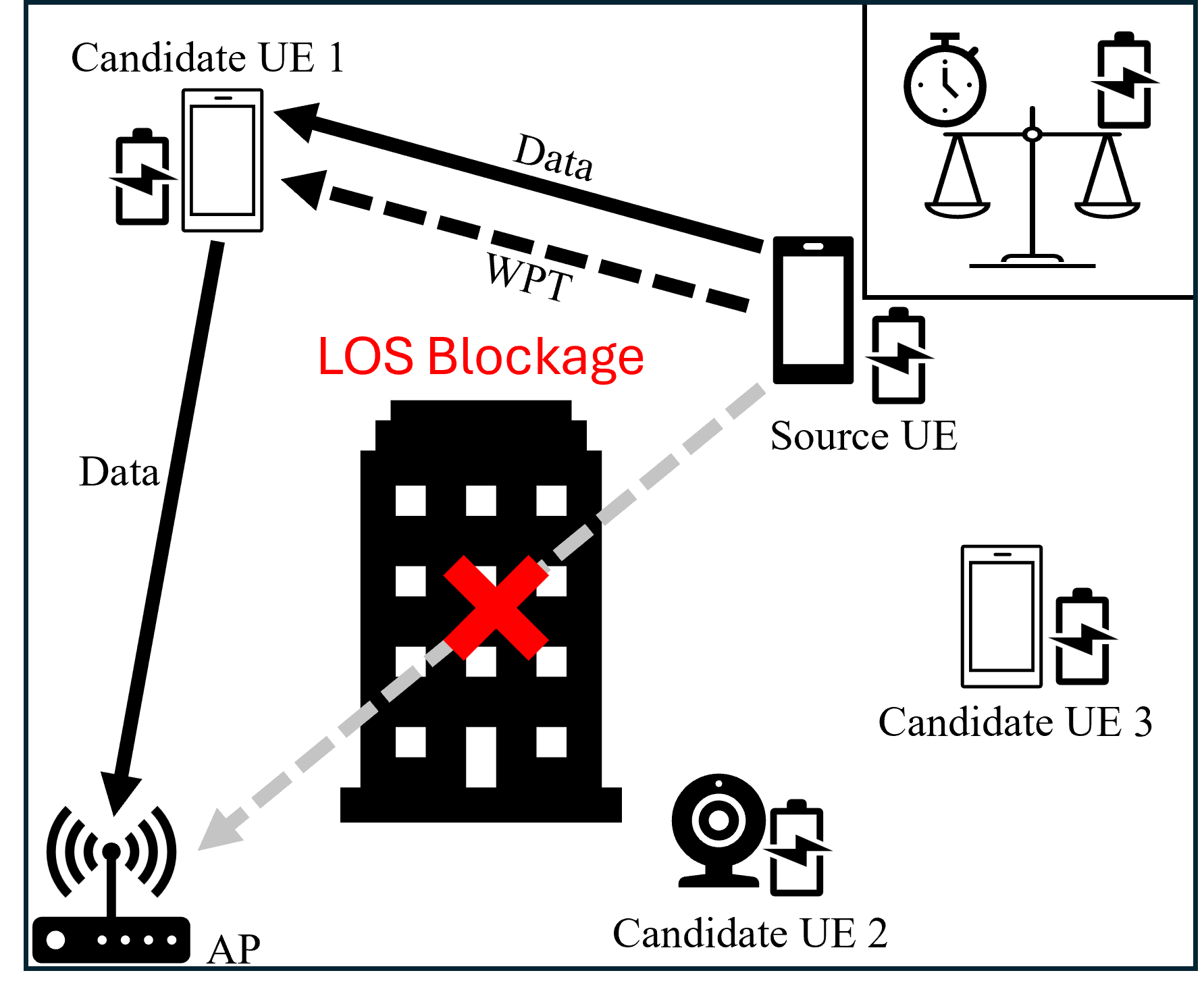}
    \vspace{-0.05in}
    \caption{Overview of the system under consideration. The source must transfer data to the AP while considering both timeliness and energy efficiency. Due to, e.g., a LOS blockage, quality of the source-to-AP channel is  poor, and the source may try to incentivize other nearby UEs to act as relays through payments of WPT.}
    \label{fig:overview}
    \vspace{-0.2in}
\end{figure}

Our work differs from the existing literature in several ways. An extensive body of work focuses on WPT-enabled relaying in a cooperative setting \cite{TWC2014Ding, IoTJ2021LiWangLiuLiPengPiranLi, Ashraf2021Access}, in contrast to the non-cooperative setting we consider. These papers analyze and optimize the source and relays as a single unit, applying game theoretical approaches only to handle incomplete information \cite{TWC2014Ding}. On the other hand, several works propose the use of auctions in the context of providing WPT as a service \cite{Ni2019ICCT,Cheng2023TNSM}. In those papers, bidders buy energy (transferred via WPT) with monetary payments, but in this work, a payment of energy purchases a relaying service.

In our recent work~\cite{hurst2024GLOBECOM}, we consider a similar setting with a hard time constraint, i.e., the source must offload the data within a certain time. This paper relaxes the time constraint, expressing the need for timely data transfer as a linear penalty term instead. This better models scenarios such as data uploading in sensor networks for non-real-time applications or asynchronous messaging. In this scenario, bids are expressed as pairs of values -- the time and the power to be used for transmission, and the resulting multi-attribute auction \cite{CheRAND1993, david2022AAMAS} requires significantly different analysis than our past work. While some of our analysis mirrors existing results such as \cite{CheRAND1993}, to our knowledge the structure of the utility functions we consider are not covered by any general theoretical works. We next state our contributions:

\begin{itemize}[leftmargin=0.15in]
    \item We present the problem of inducing efficient data offloading in a non-cooperative relaying scenario, considering both timeliness and energy consumption. The source must select a relaying candidate along with the transmission time and power, which together determine the energy payment to the candidate, while the candidates hold private information about the quality of their downlink channels to the AP. The source seeks to minimize the weighted sum of its energy consumption and the time required to complete the data transfer, while each relaying candidate seeks to maximize its net change in energy.
    \item For the non-cooperative setting, we propose a multi-attribute auction framework, and for comparison, we develop a cooperative baseline, which requires solving a mathematical optimization problem corresponding to finding the UEs' truthful bids. We prove that this problem is convex and derive a closed-form solution. 
    \item For a second-preferred-offer auction \cite{Desgagne1988Mimeo}, we prove that truthful bidding constitutes a Nash Equilibrium, but is not a dominant strategy (i.e., the auction is not incentive compatible). We then propose a modified auction, which we prove to be incentive compatible.    
    \item Finally, we present the results of several numerical simulations which validate the efficacy of our approach. These illustrate the impact of the number of candidates and the delay weighting factor on the timeliness and energy efficiency of the system as well as the net energy harvested by a candidate. Overall, our results show that the performance of the proposed approach in the non-cooperative settings approaches the performance of the cooperative baseline as the number of candidates increases.
\end{itemize}


\section{Modeling}
\vspace{-0.05in}
In this section, we present models for communication and WPT, as well as utility functions for the relevant entities. We consider a 2D coordinate system, where, without loss of generality, the access point (AP) is located at the origin, the source is located at a point, $q_s$, and $n$ candidate UEs are located at positions $\{q_i\}_{i\in\mathcal{N}}$, with $\mathcal{N} = \{1,...,n\}$. 

\subsection{Channel Modeling and Communication}
\vspace{-0.05in}
Let $H_s$, $H_i$, and $H_{s,i}$ represent the channel power between the source and the access point (AP), between a candidate and the AP, and between the source and a candidate, respectively. We assume that on any given link, both the transmitter and receiver have accurately estimated these coefficients. Specifically, the source knows $H_s$ and $H_{s,i}$ for all $i \in \mathcal{N}$, while candidate $i$ knows $H_{s,i}$ and $H_i$. However, the source does not have access to $H_i,\,\forall i \in \mathcal{N}$, and each candidate lacks knowledge of $H_s$ and $H_j,\,\forall j \in \mathcal{N} \setminus \{i\}$.

To quantify channel capacity, we apply the Shannon-Hartley theorem. The maximum achievable spectral efficiency (in bits/s/Hz) over the channel $l \in \{s, i, si\}$ is given by $r_{\text{max}} = \log_2\left(1+P_{\text{TX}}H_l/\sigma^2\right)$, where $P_{\text{TX}}$, $H_l$, and $\sigma^2$ denote the transmit power, channel power, and noise power at the receiver, respectively. It is further assumed that both the source and the candidates have a maximum transmit power of $P_{\text{max}}$. Thus, for a spectral efficiency $r$ to be feasible, it must satisfy $r \leq \log_2\left(1+P_{\text{max}}H_l/\sigma^2\right)$.

Let $D$ denote the bandwidth-normalized amount of data (in bits/Hz) the source seeks to offload. If the source offloads the data within a given amount of time, $T$, the minimum transmit power is $P_s(T) = (2^{D/T}-1)\sigma^2/H_s$, and the minimum total energy consumption is $T P_s(T)$. However, if $H_s$ is small due to, e.g., a LOS blockage between the source and the AP, $P_s(T)$ may be quite large or even exceed $P_{\text{max}}$, and relaying the data through a candidate may present a better option.

\subsection{Wireless Power Transfer}
\vspace{-0.05in}
If the source offloads data by relaying through candidate $i$, two transmissions occur: first, the source sends the data to the candidate; then the candidate sends the data to the AP. The latter requires that the candidate transmits with minimum power $P_i(T) = (2^{D/T}-1)\sigma^2/H_i$ for a given time frame, $T$. As we model the candidates as non-cooperative agents, they must be compensated for this energy expenditure, and when the source sends the data to the candidate, it must use enough power so that a portion of it may be harvested and reused.

The power harvested by the $i^\text{th}$ candidate, $P_{\text{harv},i}$, depends on the channel between the source and the candidate, $H_{s,i}$, the receiver antenna’s effective aperture, $A_r$, and the efficiency of the energy-harvesting circuitry, $\alpha$. We assume power-slitting WPT \cite{Zhou2012GLOBECOM}, with power $P_\text{WPT}$ dedicated for energy harvesting, so that the harvested power can be expressed as $P_{\text{harv},i} = H_{s,i}A_r\alpha P_{\text{WPT}} = \Tilde{\alpha}_iP_{\text{WPT}}$, where $\Tilde{\alpha}_i = H_{s,i}A_r\alpha$ represents the end-to-end WPT efficiency. 

 Let $P_{s,i}(T) = (2^{D/T}-1)\sigma^2/H_{s,i}$ denote the minimum power required to send data from the source to the $i^{\text{th}}$ candidate within time $T$, and let $P_{\text{tot}}$ denote the source's total transmit power. Assume the source sends data through candidate $i$ over time $T$, and let $P_\text{tot}^{i,\text{min}}(T) =  P_{si}(T) + P_i(T)/\tilde{\alpha}_i$. If $P_\text{tot} = P_\text{tot}^{i,\text{min}}(T)$, the $i^{\text{th}}$ candidate can harvest exactly enough energy to cover its transmission to the AP with no net energy harvested. If $P_\text{tot} > P_\text{tot}^{i,\text{min}}(T)$, the candidate harvests a net energy of $\tilde{\alpha}(P_\text{tot} - P_\text{tot}^{i,\text{min}}(T))$. If $P_{s,i}(t) \leq P_\text{tot} < P_\text{tot}^{i,\text{min}}(T)$, the candidate experiences a net loss of $\tilde{\alpha}(P_\text{tot}^{i,\text{min}}(T) - P_\text{tot} ) $. We assume $P_\text{tot} \geq P_\text{s,i}(T)$, as otherwise data transfer from the source to the candidate is unsuccessful. 

\vspace{-0.1in}
\subsection{Utilities and Preferences}
When offloading data to the AP, the source faces an inherent trade-off between the timeliness of the data transmission and energy efficiency, since offloading data at a faster rate is achievable only if more energy is used. Rather than consider a hard time constraint as in \cite{hurst2024GLOBECOM}, we model this trade-off by defining the source's utility function as follows:
\vspace{-0.05in}
\begin{equation}\label{eq:source_cost}
    U_s = - T(\lambda + P_{\text{tot}}),
    \vspace{-0.05in}
\end{equation}
where $\lambda$, which we refer to as the \textit{delay power}, is a weighting factor on the time it takes to complete data offloading. If data must be offloaded quickly, \textit{e.g.}, for safety critical uses, $\lambda$ is large, but if there is no pressing need for data transfer, $\lambda$ is small. On the other hand, each candidate's utility is its net change in energy. We assume that the transmission is sufficiently directional so that a candidate harvests energy only if it is chosen to relay. Let $\rho$ indicate the index of the chosen candidate. The utility of candidate $i$ is given by 
\begin{equation*}
    U_i = \begin{cases}
        0 & \text{if } \rho \neq i\\
        T\tilde{\alpha_i}\left(P_\text{tot} - P_\text{tot}^{i,\text{min}}(T)\right)& \text{if } \rho = i
    \end{cases},\quad \forall i \in \mathcal{N}.
\end{equation*}
We next models this scenario with a multi-attribute auction.

\section{An Auction-based Protocol}\label{sec:auction_base_protocol}
In this paper, we consider the problem of designing a mechanism which results in efficient data offloading from the source to the AP. This protocol must handle the non-cooperative nature of the candidates and decentralized decision making under incomplete information. In this section, we first present a multi-attribute auction framework which allows us to reason about the behavior of the source and candidates in a rigorous way. We then introduce a perfect information baseline and an accompanying core optimization problem relevant in our subsequent discussion.
\vspace{-0.1in}
\subsection{Multi-attribute Auctions}
\vspace{-0.05in}
Much of auction theory focuses on auctions in which bids take the form of a single value, usually a monetary payment, but in some scenarios, such as the provisioning of government contracts~\cite{CheRAND1993} or cloud computing services \cite{sharghivand2021Access}, the auctioneer must consider other aspects of a bid (e.g., the quality of a service provided) in addition to the cost. In these \textit{multi-attribute auctions}, a bid is a vector where each entry corresponds to a value for one of the attributes. The auctioneer has a \textit{scoring function}, which maps the vector-valued bids to a scalar score. This scoring function may be the same as the auctioneer's utility function. As in single-attribute auctions, a \textit{mechanism} defines the rules of the auctions, i.e., who wins and the final values of each attribute.

This multi-attribute auction framework captures the important features of our WPT-incentivized relaying scenario. We specifically consider a reverse Vickrey multi-attribute auction, so that the bid giving the lowest score wins, but the final payment corresponds to the second-lowest bid. Within the context of the communication system, the auction proceeds as follows:

\begin{enumerate}[leftmargin=0.2in]
    \item The source calculates its maximum utility if it tries to communicate directly (i.e., its \textit{reservation price}). We denote the corresponding power-time pair as $b_0 = (T_0^*, P_{\text{tot},0}^*)$.
    \item The source broadcasts a request for bids. The request includes the value of $\lambda$. 
    \item The candidates submit bids in the form of time-power pairs, $b_i = (T_i, P_{\text{tot},i}),\,\forall i \in \mathcal{N}$.
    \item The source uses the utility function, $U_s$, to score the bids, and the winner is $i^* = \argmin_{i\in\mathcal{N}\cup \{0\}} U_s(T_i, P_{\text{tot},i})$. For ease of exposition, we ignore the possibility of ties, which can be handled trivially. 
    \item The source transmits with choices of $T$ and $P_\text{tot}$ that satisfy $U_s(T, P_\text{tot}) = U_s(T_{\hat{i}}, P_{\text{tot},{\hat{i}}})$, where $\hat{i} = \min_{i\neq i^*} U_s(T_i, P_{\text{tot},i})$.
\end{enumerate}

Note that multiple values of $(T, P_\text{tot})$ satisfy $U_s(T, P_\text{tot}) = U_s(T_{\hat{i}}, P_{\text{tot},{\hat{i}}})$, but for the auction to be completely defined, we must give specific values. This adds additional complexity to multi-attribute auctions when compared to typical auctions in which bids are a single value, as in \cite{hurst2024GLOBECOM}. In the next section, we present two specific choices for these final payments and mathematically characterize their properties. However, we first present a perfect information baseline and introduce a core optimization problem that is central to later analysis.

\subsection{Perfect Information Baseline}
\vspace{-0.05in}
To better understand the challenges and limits of performance in this system, we consider a perfect-information baseline, in which the source knows the value of $H_i$ for each of the candidates, either via an oracle or via the candidates themselves in a cooperative setting. To facilitate discussion, let $z_i$ denote an end-to-end channel quality in an abstract sense, with $z_0 = \sigma^2/H_s$ and $z_i = \sigma^2(1/H_{si} + 1/(\tilde{\alpha}_iH_i)),\,\forall i \in \mathcal{N}$. In this scenario, the source can find the value of $T$ and $P_{\text{tot}}$ for each candidate, and for itself, that maximizes the source's utility, $U_s$, while accounting for the agents' rationality, i.e., ensuring that $U_i \geq 0$. Formally, this problem is expressed as:
\vspace{-0.05in}
\begin{mini}
    {T, P_{\text{tot}}}
    {T(\lambda + P_{\text{tot}})}
    {}
    {}
    \addConstraint{P_{\text{tot}} \geq P_{\text{tot}}^{i,\text{min}}(T)}
    \addConstraint{P_{\text{max}} \geq P_{\text{tot}}} 
    \addConstraint{T >0},
\end{mini}
where the first constraint ensures that the data can be transferred within the chosen time while ensuring no energy loss for the candidate, and the second constraint ensures that the maximum transmit power is respected. Since the objective is increasing in $P_{\text{tot}}$, the first constraint is active for the optimal values of $T$ and $P_{\text{TX}}$. This allows us to restate the problem as
\begin{mini}
    {T}
    {T\left(\lambda + (2^{D/T}-1)z\right)}
    {\label{prob:min_tau}}
    {T^*(z) = }
    \addConstraint{T > D/\log_2\left( 1+P_{\text{max}}/z\right)},
\end{mini}
where we use the generic argument, $z$, for ease of exposition in the discussion that follows. Solving (\ref{prob:min_tau}) gives the optimal transmit time, $T^*(z)$, the optimal transmit power, $P_{\text{tot}}^*(z) = (2^{D/T^*(z)}-1)z$, and the resulting utility for the source, $v(z)= T^*(z) (\lambda +  P_{\text{tot}}^*(z) )$. We next derive a closed-form expression for the solution to (\ref{prob:min_tau}). To do so, we first prove its convexity.
\begin{lemma}\label{lma:min_tau_convex}
    Problem (\ref{prob:min_tau}) is convex.
\end{lemma}
\begin{proof}
    Let $g(T) = T(\lambda + (2^{D/T}-1)z)$. Then we have
    \begin{equation} \label{eq:dg_dtau}
    \begin{split}
        \frac{d g}{d T}&= z 2^{D/T} \left(1 - \frac{D\ln{2}}{T}\right) + \lambda - z\\
        \frac{d^2 g}{d T^2}& = z \frac{(D \ln{2})^2}{T^3} 2^{D /T} > 0,\;\; \forall T >0,\;\forall z> 0.
    \end{split}
    \end{equation}
    Furthermore, the linear constraint ensures that $T >0$. Thus, the objective is convex over the domain of the problem.
\end{proof}

As we have proved the convexity of the problem, either the constraint is tight, i.e., $T^* = D(\log_2\left( 1+P_{\text{max}}/z\right)^{-1}$ or $T^*$ satisfies $dg(T^*)/dT = 0$. Assuming the latter, we have
\begin{equation*}
    \begin{split}
        &2^{D/T^*}\left(1-\frac{D\ln (2)}{T^*}\right) = 1 - \lambda/z \implies \\
        &T^* = \frac{D\ln(2)}{ W_0\left(\left(\frac{\lambda}{z} - 1\right)e^{-1}\right) +1}
    \end{split}
\end{equation*}
where $W_0(\cdot)$ is the principal branch of the Lambert W function. Importantly, $\lambda/z -1 > -1$, so that $W_0((\lambda/z -1)e^{-1})$ is always well defined. The final expression for $T^*(z)$ is
\begin{equation}\label{eq:taU_star}
    T^*(z) = \max \left(\frac{D\ln(2)}{ W_0\left(\left(\frac{\lambda}{z} - 1\right)e^{-1}\right) +1},\;\frac{D}{\log_2\left( 1+P_{\text{max}}/z\right)}\right)
\end{equation}

This problem is solved for all $i$ to find $T_i^* = T^*(z_i)$, $P_{\text{tot},i}^* = P_{\text{tot}}^*(z_i)$ and $v_i = v(z_i)$ for all candidates as well as the source, and the source can simply transmit for a duration $T_{i^{\dag}}^*$ with power $P_{\text{tot},i^{\dag}}^*$, where  $i^{\dag} = \argmin_i \{v_i\}_{i=0}^n$. This maximizes the sources utility while ensuring all candidates have non-negative utility, and the bid $b_i^* = (T_i^*, P_{\text{tot},i}^*)$ is the truthful bid, corresponding to revealing the private information, $H_i$. We next analyze the system performance in the non-cooperative setting with imperfect information.

\section{Analysis of Auction}\label{sec:analysis}
To fully define the auction-based protocol introduced in the Section~\ref{sec:auction_base_protocol}, the exact choice of $T$ and $P_\text{tot}$ requires careful consideration. In this section, we first consider the case where the final values of  $(T, P_\text{tot})$ are exactly $ (T_{\hat{i}}, P_{\text{tot},\hat{i}})$, sometimes referred to as a \textit{second-preferred-offer auction}~\cite{Desgagne1988Mimeo}. We then propose a slightly modified auction, which we show to be incentive compatible. In our discussion we let  $B = (b_0, b_1, ..., b_n)$ denote a \textit{strategy profile}, with $b_0 = (T_0^*, P_{\text{tot},0}^*)$, and $b_i=(T_i, P_{\text{tot}, i})$.

\subsection{The Second-Preferred-Offer Auction}
A second-preferred-offer auction is formally expressed as:
 \begin{equation}\label{eq:intuitive_mechanism}
    \begin{split}
        \rho(B) &= i^* = \argmax_{i\in\mathcal{N}\cup \{0\}}U_s(T_i, P_{\text{tot},i})\\
        X(B) &= b_{\hat{i}} = (T_{\hat{i}}, P_{\text{tot},\hat{i}}),
    \end{split}
    \end{equation}
where $\rho(B)$ gives the winning candidate (or the source) and $x(B)$ gives the source transmit time, $T$, and power, $P_{\text{tot}}$. We prove that under this mechanism, the strategy profile $B^* = (b_1^*, ..., b_n^*),\, b_i^* = (T_i^*, P_{\text{tot}, i}^*)$ constitutes a Nash Equilibrium. To do so, we first derive necessary intermediate results.

Consider (\ref{prob:min_tau}). This problem gives a general mapping between a any value of the $z$ and the corresponding optimal time $T^*(z)$. We now prove important properties of $T^*(z)$ and $P_{\text{tot}}^*(z)$, which we later use to develop our main results.  
\begin{lemma}\label{lma:increasing_in_z}
    Optimal relay time, $T^*(z)$, is increasing in $z$.
\end{lemma}
\begin{proof}
    First, note that (\ref{eq:taU_star}) is continuous for $z >0$.
    From (\ref{eq:dg_dtau}), we have $ z = -\lambda(2^{D/T^*}(1 - D\ln{2}/T^*) - 1)^{-1}$ when the constraint is inactive, and taking the derivative with respect to $T^*$, we find
    \begin{equation}\label{eq:dx_dt_inactive}
        \frac{d z}{d T^*} = \frac{\lambda (D\ln{2})^2 2^{D/T^*}}{(2^{D/T^*}(1 - D\ln(2)/T^*)-1)^2 (T^*)^3} >0.
    \end{equation}
    If the constraints is active, taking the derivative, we find
    \begin{equation*}
        \frac{d T^*}{d z} = \frac{DP_{\text{max}}}{\left(\log_2(1+P_{\text{max}}/z)\right)^2 z \ln(2)(P_\text{max} + z)} >0
    \end{equation*}
    Thus, $T^*(z)$ is increasing in $z$, and we can further show that the optimized power, $P_{\text{tot}}^*(z)$, is increasing in $z$ as well.
\end{proof}

\begin{lemma}\label{lma:P_tot_in_z}
    The optimal transmit power $P_{\text{tot}}^*(z) = (2^{D/T^*(z)}-1) z$ is non-decreasing in $z$.
\end{lemma}
\begin{proof}
     Note that in the regime where the constraint is active, $P_{\text{tot}}^*(z) = P_{\text{max}}$, and when the constraint is inactive, $P_{\text{tot}}^*(z) < P_{\text{max}}$. Furthermore, when the constraint is inactive,
    \begin{equation*}
        \frac{d P_{\text{tot}}^*(z)}{d z} = (2^{D/T^*(z)}-1) + z\left(\frac{-D 2^{D/T^*(z)} \ln{2}}{T^*(z)^2}\right)\frac{dT^*}{dz}
    \end{equation*}
    Using (\ref{eq:dx_dt_inactive}) and (\ref{eq:dg_dtau}), we find that 
    \vspace{-0.1in}
    \begin{equation*}
        z\frac{dT^*(z)}{dz} = \frac{-(2^{D/T}(1-\frac{D\ln(2)}{T} -1)T^3}{(D\ln(2))^2 2^{D/T}},
        \vspace{-0.05in}
    \end{equation*}
    which gives
    \begin{equation*}
    \begin{split}
       \frac{d P_{\text{tot}}^*(z)}{d z} & = \frac{T}{D\ln(2)} (e^{D\ln(2)/T}-1 - \frac{D\ln(2)}{T}) 
    \end{split}
    \vspace{-0.1in}
    \end{equation*}
    Letting $\omega = \frac{D\ln(2)}{T}$, we have $d P_{\text{tot}}^*(z)/d(z) = (1/\omega)(e^{\omega} - \omega - 1 > 0,\;\forall \omega > 0$.  Finally, note that $P_{\text{tot}}^*(z)$ is continuous even at points where $T^*(z)$ is not differentiable.
    \vspace{-0.05in}
\end{proof}

As a consequence, the optimal source utility as a function of $z$, $v(z) = T^*(z)(\lambda + P_{\text{tot}}^*(z))$ is increasing in $z$. We next use this to prove a central result underpinning our analysis.

\vspace{-0.05in}
\begin{theorem}\label{thrm:manifold}
    If all bids are restricted to the set described by solutions to (\ref{prob:min_tau}), i.e., $b_i\in\mathcal{B} = \{(T^*(z), P_{\text{tot}}^*(z)) | z > 0\},\,\forall i \in \mathcal{N}$, then the strategy $b_i^* = (T_i^*, P_{\text{tot},i}^*)$ is a dominant strategy.
\end{theorem}
\begin{proof}
    First, let $\tilde{z}_i$ satisfy $b_i =(T_i, P_{\text{tot},i}) = (T_i^*(\tilde{z}_i), P_{\text{tot},i}^*(\tilde{z}_i))$, and note that $\tilde{z}_i$ may differ from $z_i$. Under the mechanism described in (\ref{eq:intuitive_mechanism}), if $i^* \neq 0$, we may express the utility of the winning candidate as $U_{i^*} =            T_{\hat{i}}\tilde{\alpha}_i(2^{D/T_{\hat{i}}}-1)(\tilde{z}_{\hat{i}} - z_{i^*}),$
        while for all other candidates, $U_i = 0$.
        Thus, if $\tilde{z}_i > z_{i^*}$, the winning candidate experiences a net gain in energy, while if $\tilde{z}_i < z_{i^*}$, the candidate experiences a net loss in energy. Further note that, from Lemmas~\ref{lma:increasing_in_z} and \ref{lma:P_tot_in_z}, $\rho(B) = \argmin_i \tilde{z}_i$. Thus, if $\tilde{z}_{i^*} > z_{i^*}$, the winner may not win in cases where it could have won and received a positive utility, and if $\tilde{z}_{i^*} < z_{i^*}$, candidate $i^*$ may actually experience a net loss of energy upon winning.
\end{proof}

Using the above result, we can now state our desired result:
\begin{corollary}\label{crly:NE}
    Under the mechanism in (\ref{eq:intuitive_mechanism}), the strategy profile $B^* = \{b_i^*\}_{i\in\mathcal{N}}$, $b_i^* = (T_i^*, P_{\text{tot},i}^*)$, is a Nash Equilibrium.
\end{corollary}
\begin{proof}
       This follows from the fact that $b_i^* \in \mathcal{B}$ and Theorem~\ref{thrm:manifold}.
\end{proof}

Thus, under the mechanism given in (\ref{eq:intuitive_mechanism}), none of the candidates have an incentive to deviate from bidding $(T_i^*,P_{\text{tot},i}^*)$ as long as all other candidates are following the same strategy. However, this does not tell us how a candidate will act if the other candidates are not using bids $b_i^*$, and in fact, a bid of $b_i^*$ may result in a net loss in energy. In particular, a winning candidate could lose energy if there were a $T$ and $P$ such that $U_s(T, P) > U_s(T_{i^*}^*, P_{\text{tot},i^*}^*)$ but such that $P - P_{\text{tot}}^{i^*,\text{min}}(T) < 0$.

To see that such a choice is possible, note that $P_{\text{tot}}^{i^*,\text{min}}(T) = (2^{D/T} - 1)z_{i^*}$, and let $P = (2^{D/T} - 1)(z_{i^*} - \epsilon_1) < P_{\text{tot},i}^*$ for $z_{i^*} > \epsilon_1 > 0$. Now let $T = T_i^* + \epsilon_2$ for $\epsilon_2 > 0$, and note that when $\epsilon_2 > T_i^*\epsilon_1/\left((z_i - \epsilon_1) + \lambda\right)$, we have $U_s(T, P) > U_s(T_i^*, P_{\text{tot},i}^*)$. Thus, $b_i^*$ is not a dominant strategy, and the mechanism is not incentive compatible. We next present a modification to (\ref{eq:intuitive_mechanism}) which makes $b_i^*$ a dominant strategy.

\subsection{A Modified Auction Mechanism}
\vspace{-0.05in}
As an alternative to a second-preferred-offer auction, we propose a modification that maps the second lowest bid into the set $\mathcal{B}$. Let $z(u)$ satisfy $c=U_s(T^*(z), P_{\text{tot}}^*(z))$, that is, $z(u)$ is the inverse of $v(z)$, and this inverse is well defined from Lemmas~\ref{lma:increasing_in_z} and \ref{lma:P_tot_in_z}. The modified mechanism is:
\begin{equation}\label{eq:incentive_compatible_mechanism}
        \begin{split}
        \rho(B) &= i^* = \argmax_{i\in\mathcal{N}\cap \{0\}} U_s(T_i, P_{\text{tot},i})\\
        X^{\text{IC}}(B) &= \left(T^*\left(z\left(c_{\hat{i}}\right)\right), P_{\text{tot}}^*\left(z(u_{\hat{i}})\right)\right),
    \end{split}
    \end{equation}
with $u_{\hat{i}} = U_s(T_{\hat{i}}, P_{\text{tot},\hat{i}})$. By mapping the bid corresponding to the second-lowest score into the set of bids $\mathcal{B}$, we effectively satisfy the condition of Theorem~\ref{thrm:manifold}.
\begin{corollary}
    Under the auction described in (\ref{eq:incentive_compatible_mechanism}), bidding $b_i^*$ is a dominant strategy for all $i\in\mathcal{N}$.
\end{corollary}
\begin{proof}
    Note that (\ref{eq:incentive_compatible_mechanism}) maps $b_{\hat{i}}$ into $\mathcal{B}$, and the result follows from Theorem~\ref{thrm:manifold}.
\end{proof}
Thus, our modified auction is incentive compatible. Due to this property, the behavior of the rational agents in system is well defined, and we next provide a series of numerical results that validate the efficacy of this proposed approach.

\begin{figure}
    \centering
    \includegraphics[width=0.9\linewidth, trim={0.1in, 0.1in, 0.1in, 0.05in}, clip]{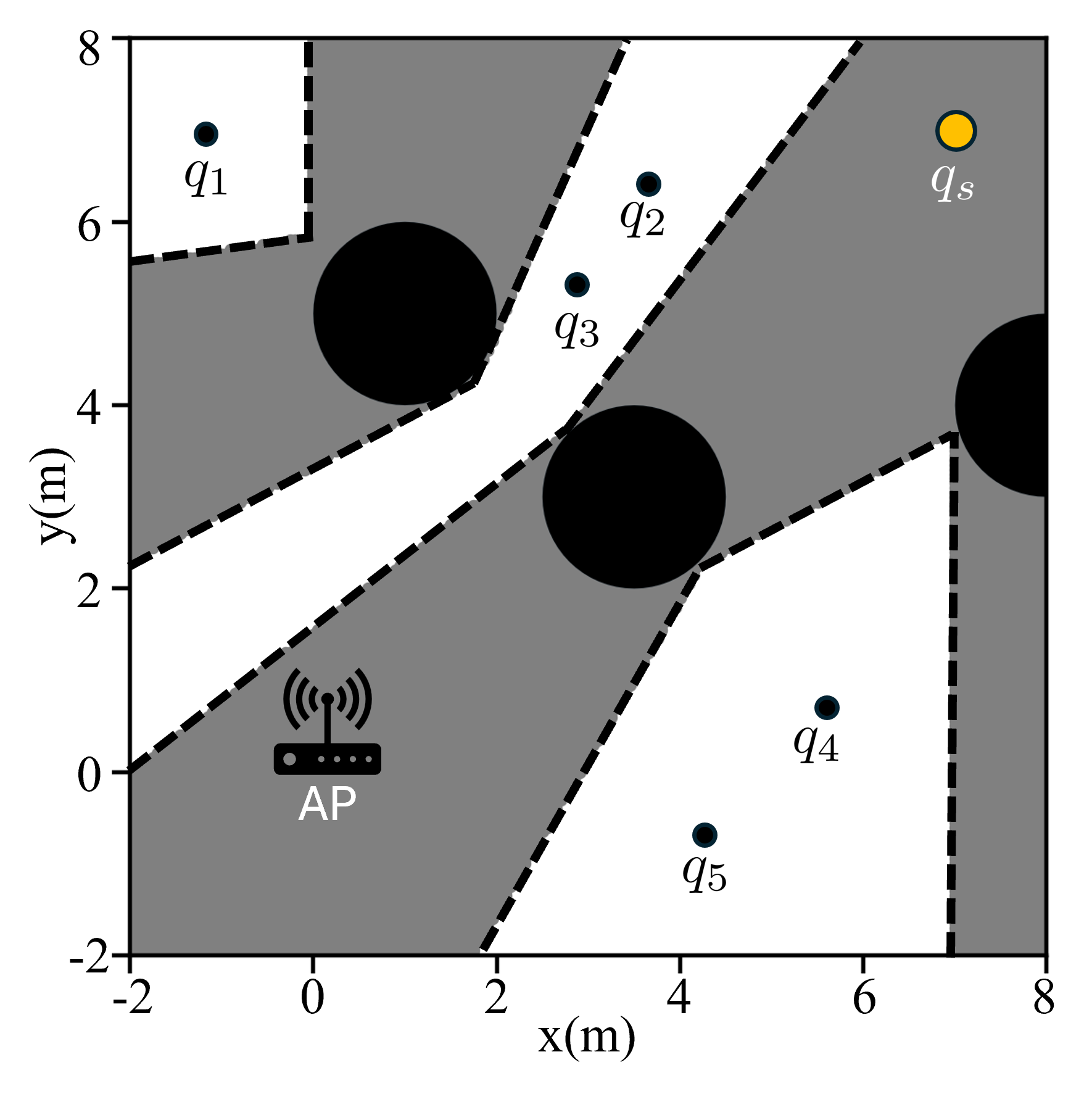}
    \vspace{-0.1in}
    \caption{Setting of numerical results presented in Section~\ref{sec:results}. The source, labeled at $q_s$ must offload data to the AP, located at the origin, but a blockage, shown as a large black circle, creates a deep fade. The regions in grey indicate locations for which a link with the source or with the AP is non-LOS, and we therefore do not consider these as valid locations for candidates. Valid candidate placements regions are in white, and sample candidate locations, $q_i$ for $n=5$, are shown as small black circles.}
    \label{fig:setup}
    \vspace{-0.2in}
\end{figure}

\section{Numerical Results}\label{sec:results}
In this section, we present extensive numerical studies demonstrating the efficacy of our auction-based WPT-enabled relaying protocol. To illustrate this proposed approach, we consider the scenario depicted in Fig.~\ref{fig:setup}, in which a blockage sits on the LOS path between the source and the AP, creating a deep fade. We model the non-LOS channel power between the source and the AP using the same approach as in \cite{Hemadeh2018CST}. Specifically, we have $H_s = K^{\text{NLOS}} d_s^{-\eta^{\text{NLOS}}} H_{s,f}$, where $K^{\text{NLOS}}$ and $\eta^{\text{NLOS}}$ are the path loss intercept and path loss exponent, respectively,  for a non-LOS channel, and $d_s=||q_s||_2$ is the distance between the source and the AP. Additionally, we assume Rayleigh fading, so that the channel fading power, $H_{s,f}$, is exponentially distributed, with $\lambda = 1$ \cite{GoldsmithWC}. We only consider candidates that have a LOS channel to both the AP and the source, as otherwise, they would almost certainly offer no benefit, and we model these channels as $H_{i} =  K^{\text{LOS}} d_s^{-\eta^{\text{LOS}}} H_{s,f}$ and $H_{s,i} =  K^{\text{LOS}} d_{s,i}^{-\eta^{\text{LOS}}} H_{s,f}$, where  $K^{\text{LOS}}$ and $\eta^{\text{NLOS}}$ are the path loss intercept and path loss exponent, respectively,  for a LOS channel, and $d_i=||q_i||_2$ and $d_{s,i} = ||q_i - q_s||_2$. For all simulations, we use the following specific choice of communication channel parameters: $K^{\text{NLOS}}=-25$, $\eta^{\text{NLOS}} = 5.76$, $K^{\text{LOS}}=0$, $\eta^{\text{NLOS}} = 2.5$, and $q_s = (7,7)$. For the data requirement, we set $D = 8\,$bits/Hz. For WPT parameters, we use $A_r = 1\,\text{cm}^2$ and $\alpha = 0.2$.

\begin{figure*}
    \centering
    \includegraphics[width=1\linewidth]{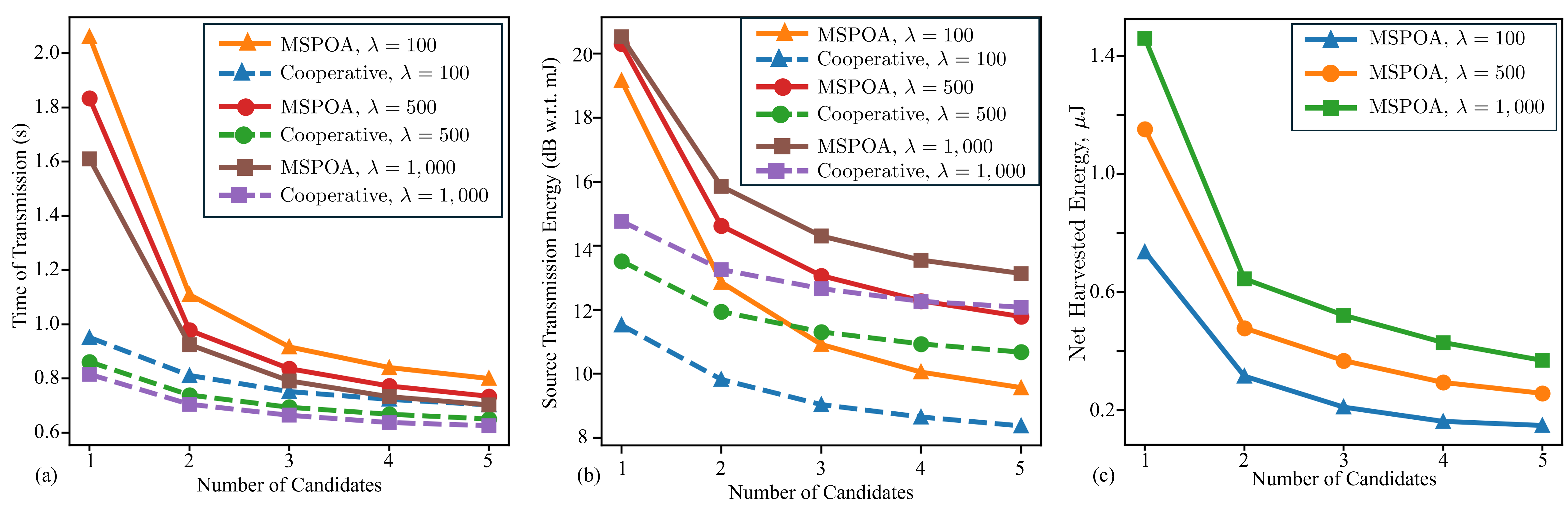}
    \caption{(a) Average transmission duration, $T$, as a function of the number of candidates, $n$. As the number of candidates increases, the transmission time decreases sharply, and greater values of delay power, $\lambda$, lead to shorter transmission times, too. (b) Average millijoules of transmission energy, in dB, as a function of the number of candidates $n$. As the number of candidates increases, the energy usage decreases, and smaller values of delay power, $\lambda$, result in greater energy efficiency. (c) Average net harvested energy by the candidate as a function of the number of candidates in the MSPOA. The harvested energy decreases as the number of candidates increases or as the delay power, $\lambda$, decreases.}
    \label{fig:results}
    \vspace{-0.1in}
\end{figure*}

For a given number of candidates, $n$, and a given delay power, $\lambda$, we run 5,000 experiments. For each experiment, we uniformly and independently sample $q_i$ from the regions with LOS paths to both the source and the AP, and then we independently sample the fading power, $H_{s,f}$, for each channel. We then run both the cooperative baseline described in Section~\ref{sec:auction_base_protocol} and our proposed approach given in (\ref{eq:incentive_compatible_mechanism}), which we refer to as the modified second-preferred-offer auction (MSPOA). When discussing system performance with respect to the choice of $n$ and $\lambda$, we average over these batches of 5,000 experiments. 

\subsection{Transmission Time}
We now examine the impact of the number of candidates, $n$, and the delay power, $\lambda$, on the transmission time, $T$. Fig.~\ref{fig:results} (a) shows that under the VAP, the transmission time decreases greatly as the number of candidates increases. For example, the value of $T$ decreases by over $55\%$ from $n=1$ to $n=3$ when $\lambda = 100$. Furthermore, as $n$ increases, the transmission time in the non-cooperative setting approaches that of the cooperative baseline. We also see that placing greater importance on the timeliness of data transmission (larger values of $\lambda$) leads to shorter transmission times.

\subsection{Energy Consumption}
We next present results on the relationship between key system parameters and and total communication energy. Fig.~\ref{fig:results} (b) again shows an acute decline in energy usage as the number of candidates increases, with as much as a $76\%$ energy reduction for as little as $n=2$ candidates in the non-cooperative setting. The energy usage in the non-cooperative scenario using our proposed approach becomes closer to the energy usage in the cooperative case as $n$ grows larger as well, and we also see that smaller values of delay power, $\lambda$, lead to greater energy efficiency.

\subsection{Energy Harvesting}
Finally, we compare the net energy harvested by a winning candidate across different values of the delay power, $\lambda$, and number of candidates, $n$. Fig.~\ref{fig:results} (c) shows that as the number of candidates increases, the surplus energy harvested beyond what is needed for the candidate to transmit to the AP decreases. We also see that as $\lambda$ increases, the harvested energy increases, indicating that if the source pays a premium for more urgent data transmissions ($\lambda$ is larger).

\section{Conclusions}
This paper demonstrates the effectiveness of multi-attribute auctions for enhancing relaying efficiency in non-cooperative communication systems. We offer a robust framework that significantly reduces energy consumption and transmission delays by strategically selecting relay candidates and optimizing transmission time and power. Our findings reveal that the proposed auction mechanism not only encourages truthful bidding as a dominant strategy but also performs close to a cooperative baseline as the number of relay candidates increases. Robust numerical validation underscores the potential of our approach to improve data offloading in challenging communication environments. This research paves the way for future investigations into auction-based strategies in wireless networks, promising enhanced performance in non-cooperative settings.

\bibliographystyle{IEEEtran}
{\footnotesize
\bibliography{main}}

\end{document}